\documentclass[12pt]{article}
\usepackage{amsmath, amsthm}
\usepackage{amssymb} 
\usepackage{url} 
\usepackage{float} 
\usepackage{graphicx}
\newtheoremstyle{theorem}
{10pt} 
{10pt} 
{\sl} 
{\parindent} 
{\bf} 
{. } 
{ } 
{} 
\theoremstyle{theorem}
\newtheorem{theorem}{Theorem}

\newtheoremstyle{defi}
{10pt} 
{10pt} 
{\rm} 
{\parindent} 
{\bf} 
{. } 
{ } 
{} 
\theoremstyle{defi}

\begin{document} 

\title{A new stochastic diffusion process to model and predict electricity production from natural gas sources in the United States}

\author{Safa' Alsheyab\\[6pt]
Department of Mathematics and Statistics\\ Jordan University of Science and Technology\\ P.O.Box 3030, Irbid 22110, Jordan\\
e-mail:smalsheyab6@just.edu.jo\\[6pt]
}

\maketitle

\begin{abstract}
This paper introduces a new stochastic diffusion process to model the electricity production from natural gas sources (as a percentage of total electricity production) in the United States. The method employs trend function
analysis to generate fits and forecasts with both conditional and unconditional estimated trend functions. Parameters are estimated using the maximum likelihood (ML) method, based on discrete sampling paths of the variable "electricity production from natural gas sources in the United States" with annual data from 1990 to 2021. The results show that the
proposed model effectively fits the data and provides dependable medium-term forecasts for 2022-2023.

\medskip

{\bf Math. Subject Classification:} 62M86, 60H30, 65C30

{\bf Key Words and Phrases:} Stochastic diffusion process; maximum likelihood estimate; fit and forecast; trend function; prediction accuracy; electricity production from natural gas

\end{abstract}

\section{Introduction}

Electricity is a fundamental part of modern life, greatly improving our standard of living by powering homes, businesses, and technologies. However, generating electricity can also harm the environment. The extent of this damage largely depends on how electricity is produced. For instance, burning coal releases twice as much carbon dioxide as burning an equal amount of natural gas. Considering the many benefits of natural gas, such as its relatively low cost and low carbon footprint, it is considered a good source of electricity. Plus, gas plants can be built quickly, unlike nuclear facilities.
The percentage of electricity produced from natural gas, $\%$ of total, represents the share of natural gas in overall
electricity generation, measured in total gigawatt hours (GWh) produced by power plants.

Electricity production in the United States has undergone a structural transformation, with natural gas becoming one of the main sources of power generation. Understanding the dynamics of its share in total electricity production is essential for policy design, risk management, and long-term planning. Electricity generation from natural gas in the United States has grown steadily over the past thirty years. As reported in Table 1, production increased from 11.85815 units in 1990 to over 41.90 units in 2023, more than tripling. However, this growth has not been steady; the data show fluctuations and irregular deviations. For example, while production rose steadily during the early 2000s, it experienced a temporary dip in 2003, a sharp acceleration after 2009, and more recently, a slight decline in 2021 before rebounding in 2022–2023. Such variability indicates the influence of both long-term growth trends and short-term stochastic disturbances. Modelling such behaviour requires a framework that captures both deterministic dynamics and random shocks. Stochastic differential equations (SDEs) provide a flexible approach that combines deterministic trends with random variability to meet these challenges.\

Previous studies have extensively used stochastic models in finance and natural sciences, with growing applications in energy markets. For example, Schwartz \cite{Schwartz1997} developed stochastic models for commodity prices, which were later adapted for electricity markets. Geman and
Roncoroni \cite{GemanRoncoroni2006} highlighted the importance of capturing seasonality and jumps in power prices, Benth et al. \cite{Benth} used stochastic processes to describe electricity spot prices with periodic components, while  Weron \cite{Weron2014} provided a thorough review of stochastic models designed specifically for power markets, emphasizing the need to include both cyclical and random dynamics to accurately represent electricity price behavior.
 In the context of electricity, up to our knowledge, most of the literature focuses on prices, or consumption, or optimization models, or uses different methodological approaches such as stochastic time change, hybrid diffusions, or optimization frameworks. For example, Borovkova and Schmeck \cite{Borovkova} have primarily focused on modelling electricity prices using a time-changed jump diffusion that adapts volatility and jump intensity over time. Mert et al. \cite{Mert} develop a hybrid stochastic diffusion ensemble specifically for forecasting natural-gas prices. Leonel et al. \cite{Leonel} focus on optimization-based decision-making for large industrial self-producers using natural gas, aiming to minimize energy costs under uncertainty through a two-stage stochastic optimization framework. While for the natural gas production in the United States, Cai and Deng \cite{Cai} propose a grid-optimized fractional grey model to forecast U.S. natural gas production; their focus is on total production forecasting using deterministic/grey system techniques.\\
  Despite these advancements, limited attention has been given to electricity production shares by fuel source within a stochastic diffusion framework. This creates a gap in the literature. 
Our goal is to use a new stochastic differential equation model to analyze U.S. electricity production from natural gas between 1990 and 2021 and to forecast for 2022-2023. The stochastic component captures external shocks and random variability.
Building on this foundation, we consider the following SDE of the form
\begin{equation}
dX(t)=\left(\frac{2}{t}-\lambda-\frac{\lambda\pi}{2 t^2} e^{-\lambda/ t}\cot(\frac{\pi}{2 } e^{-\lambda/ t})\right)X(t)dt+\sigma X(t)dW(t),
\end{equation}  
 which is particularly relevant for modelling U.S. electricity production from natural gas ($\%$ of total), here $W$ is the standard Brownian motion, and $\lambda$ and $\sigma$ are unknown indexed parameters. For more information about Brownian motion, SDEs, and their applications, we refer the reader to \cite{Karatzas}, \cite{Evans}, \cite{Oksendal}, \cite{Mao}, and \cite{Sobczyk}.

 The remainder of this paper is organized as follows: In the subsequent section, the proposal model is defined as the solution to a stochastic differential equation (SDE), from which the explicit expression of the process, the transition probability distribution function (TPDF), and the moments( particularly, the conditional and unconditional trends of the process) are derived. Section \ref{MLsection} is dedicated to estimating the model's parameters using the Maximum Likelihood Estimation (MLE) method, based on discrete sampling of the process. In Section \ref{Application}, the application of the proposed model to specific data is presented: electricity production from natural gas sources ($\%$ of total)  in the United States for the period 1990-2023. The accuracy of the forecast is assessed using standard performance metrics, including mean absolute error (MAE), root mean square error (RMSE), and mean absolute percentage error (MAPE). The performance of the proposed model is compared with that of the Gompertz model using the same dataset. Furthermore, a small Monte Carlo experiment is conducted to assess the robustness of the estimation procedure.  Finally, Section \ref{Conclusion} offers concluding remarks.

\section{The model and its  probabilistic characteristics}
In this section, we present a novel Sine-Like stochastic diffusion process. We derive and analyze the critical properties of this process, encompassing the existence of solutions, transition probability distributions, and moments.
\subsection{Sine-Like (SL) stochastic diffusion model}
The SL process is the non-homogeneous diffusion process $\{X(t),t \in [t_1,T], t_1>0\}$ with values in $(0,\infty)$, and satisfies the following It$\hat{\text{o}}$'s SDE
\begin{equation}\label{SDE2}
dX(t)=a(t,X(t),\theta)dt+b^{1/2}(t,X(t),\theta)dW(t);\hspace{.5cm} X(t_1)=X_1
\end{equation}  
with infinitesimal moments given by 
\begin{align}\label{infinitesimal}
&a(t,X(t),\theta)=\left(\frac{2}{t}-\lambda-\frac{\lambda\pi}{2 t^2} e^{-\lambda/ t}\cot(\frac{\pi}{2 } e^{-\lambda/ t})\right)X(t),\nonumber\\
&b^{1/2}(t,X(t),\theta)=\sigma X(t)
\end{align}
where $\sigma>0$, and $\lambda$  is a non-zero real constant. The solution of \eqref{SDE2}-\eqref{infinitesimal} is found, by applying It$\hat{\text{o}}$'s integral, as follows
\begin{align*}
X(t)=X_{1}+&\int_{t_{1}}^{t}\left(\frac{2}{s}-\lambda-\frac{\lambda\pi}{2 s^2} e^{-\lambda/ s}\cot(\frac{\pi}{2 } e^{-\lambda/ s})\right)X(s)ds\\
&+\sigma\int_{t_{1}}^{t}X(s)\text{d}W(s)
\end{align*}
\subsection{Existence and uniqueness}
In this part, we show the existence and uniqueness of the solution for the SL process given in \eqref{SDE2}-\eqref{infinitesimal}. To achieve this goal, it is sufficient to verify uniform Lipschitz and linear growth conditions for the infinitesimal moments; see \cite{Kloeden}, which are realized in the following theorem.
\begin{theorem}\label{th1}
The SDE in \eqref{SDE2}-\eqref{infinitesimal} has a unique solution. 
\end{theorem}
\begin{proof}
On the one hand, consider $x,y\in\mathbb{R}^+$ and $t\in[t_1,T]$. It then follows that
\begin{align*}
&\vert a(t,x)-a(t,y)\vert+\vert \sqrt{b(t,x)}-\sqrt{b(t,y)}\vert\\
&=\vert a(t,x-y)\vert+\vert \sqrt{b(t,x-y)}\vert,\\
&=\left\vert\left(\frac{2}{t}-\lambda-\frac{\lambda\pi}{2 t^2} e^{-\lambda/ t}\cot(\frac{\pi}{2 } e^{-\lambda/ t})\right)(x-y)\right\vert+\vert {\sigma(x-y)}\vert,\\
&=\left(\left\vert\left(\frac{2}{t}-\lambda-\frac{\lambda\pi}{2 t^2} e^{-\lambda/ t}\cot(\frac{\pi}{2 } e^{-\lambda/ t})\right)\right\vert+\vert\sigma\vert\right)\vert x-y\vert,\\
&\leq\left(\underset{t_1\leq t\leq T}{\sup}\left\vert\left(\frac{2}{t}-\lambda-\frac{\lambda\pi}{2 t^2} e^{-\lambda/ t}\cot(\frac{\pi}{2 } e^{-\lambda/ t})\right)\right\vert+\vert\sigma\vert\right)\left\vert x-y\right\vert
\end{align*}
therefore, the SL process satisfies a uniform Lipschitz. On the other hand, this process satisfies linear growth condition as for $y=0$, we have
\begin{align*}
&\vert a(t,x)\vert^2+\vert \sqrt{b(t,x)}\vert^2\leq\left(\vert a(t,x)\vert+\vert \sqrt{b(t,x)}\vert\right)^2,\\
&\leq\left[\left(\underset{t_1\leq t\leq T}{\sup}\left\vert\left(\frac{2}{t}-\lambda-\frac{\lambda\pi}{2 t^2} e^{-\lambda/ t}\cot(\frac{\pi}{2 } e^{-\lambda/ t})\right)\right\vert+\vert\sigma\vert\right)\vert x\vert\right]^2,\\
&\leq\left(\underset{t_1\leq t\leq T}{\sup}\left\vert\left(\frac{2}{t}-\lambda-\frac{\lambda\pi}{2 t^2} e^{-\lambda/ t}\cot(\frac{\pi}{2 } e^{-\lambda/ t})\right)\right\vert+\vert\sigma\vert\right)^2(1+\vert x\vert)^2.
\end{align*}
Thus, there is an almost surely (a.s.)  continuous process $\{x(t),t\in[t_1,T];t_1>0\}$ that is a.s. the unique solution of the  SDE (\ref{SDE2})-\eqref{infinitesimal}. 
\end{proof}
\subsection{The probability distribution  and moments of the process}
The explicit solution of the SDE (\ref{SDE2})-\eqref{infinitesimal} can be obtained by means of the appropriate transformation; $Y(t)=\log(X(t))$, and by applying It\^{o}'s formula \cite[Theorem 4.57]{JS03} to $Y$. Then, we have the following
\begin{align*}
dY(t)&= \frac{1}{X(t)}dX(t)- \frac{1}{2X^2(t)}\sigma ^2 X^2(t)dt\\
&=\left(\frac{2}{t}-\lambda-\frac{\lambda\pi}{2 t^2} e^{-\lambda/ t}\cot(\frac{\pi}{2 } e^{-\lambda/ t})\right)dt+\sigma dW(t)- \frac{\sigma ^2 }{2}dt\\
&=\left(\frac{2}{t}-\lambda-\frac{\lambda\pi}{2 t^2} e^{-\lambda/ t}\cot(\frac{\pi}{2 } e^{-\lambda/ t})- \frac{\sigma ^2 }{2}\right)+\sigma dW(t)\\
\end{align*}
Where $Y(t_1)=\log(X_1)$. By integrating the above equation, we obtain,
\begin{align*}
Y(t)-Y(t_1)= \int_{t_1}^{t}&\left(\frac{2}{s}-\lambda-\frac{\lambda\pi}{2 s^2} e^{-\lambda/ s}\cot(\frac{\pi}{2 } e^{-\lambda/ s})- \frac{\sigma ^2 }{2}\right)ds\\
&+\sigma (W(t)-W(t_1)),
\end{align*}
and hence
\begin{align*}
Y(t)&=Y(t_1)+2 \log(t/t_1)-\lambda(t-t_1)+\log\left(\sin\left(\frac{\pi}{2} e^{-\lambda/ t}\right)\right)\\
&\hspace{1cm}-\log\left(\sin\left(\frac{\pi}{2}e^{-\lambda/ t_1}\right)\right)- \frac{\sigma ^2 }{2}(t-t_1)+\sigma (W(t)-W(t_1)),
\end{align*}
Therefore,  the solution in terms of the original SL process

\begin{align}
X(t)&=X_1\left(\frac{t}{t_1}\right)^{ 2 }\frac{\sin\left(\frac{\pi}{2} e^{-\lambda/ t}\right)}{\sin\left(\frac{\pi}{2} e^{-\lambda/ t_1}\right)}e^{-\lambda(t-t_1)- \frac{\sigma ^2 }{2}(t-t_1)}e^{\sigma (W(t)-W(t_1))}\label{exa}
\end{align}
Observe that since the initial condition $Y(t_1)$ is constant a.s., and $Y(t)$ is a Markovian process, it then follows that the finite dimensional distribution of $Y(t)$ is normal. Then, the finite dimensional distribution of $X(t)$ is log-normal distribution, and  the transition probability distribution (TPDF) of $X(t)$ given $X(s)$ where $s<t$ follows a log-normal distribution
 denoted by $\Lambda_1(\mu(s,t,x_{s}), \sigma^2 (t-s))$, where $\mu(s,t,x_s)$ is given by
\begin{align*}
\mu(s,t,x)=\log(x)&+2\log(t/s)-\lambda(t-s)+\log\left(\sin\left(\frac{\pi}{2} e^{-\lambda/ t}\right)\right)\\
&-\log\left(\sin\left(\frac{\pi}{2}e^{-\lambda/ s}\right)\right)- \frac{\sigma ^2 }{2}(t-s).
\end{align*}
Therefore, the TPDF of the process has the following form
\begin{align}\label{pdf}
f(y,t|x_s,s)=\frac{1}{y}\left[2\pi\sigma^2(t-s) \right]^{-1/2}\exp\left(-\frac{[\log(y)-\mu(s,t,x_s)]^2}{2\sigma^2(t-s)}\right).
\end{align}
Taking into account that $X(t)|X(s)=x_s$ is distributed according to $\Lambda_1(\mu(s,t,x_{ts}), \sigma^2 (t-s)), $ and taking into account the properties of this distribution, the $n^{th}$ conditional moment of $X(t)$ given $X(s)$ is
\begin{align*}
E[X^n(t)|X(s)=x_s]=\exp\left(n\mu(s,t,x_s)+\frac{n^2\sigma^2}{2}(t-s)\right)
\end{align*}
Then, the conditional mean, which is considered as the trend function $(n=1)$ of the process, is
\begin{align}\label{ECTF1}
E[X(t)|X(s)=x_s]=x_s\left(\frac{t}{s}\right)^{2}e^{-\lambda(t-s)+\log\left(\sin\left(\frac{\pi}{2} e^{-\lambda/ t}\right)\right)-\log\left(\sin\left(\frac{\pi}{2}e^{-\lambda/ s}\right)\right)}
\end{align}
On the other hand, the mean function or the unconditional trend of the process under the assumption that  $P(x(t_1)=x_1)=1$ is given by 
\begin{align}\label{ETF1}
E[X(t)]&=x_{1}\left(\frac{t}{t_1}\right)^{2}e^{-\lambda(t-t_1)+\log\left(\sin\left(\frac{\pi}{2} e^{-\lambda/ t}\right)\right)-\log\left(\sin\left(\frac{\pi}{2} e^{-\lambda/ t_1}\right)\right)}\nonumber\\
&=x_{1}\left(\frac{t}{t_1}\right)^{2}\frac{\sin\left(\frac{\pi}{2} e^{-\lambda/ t}\right)}{\sin\left(\frac{\pi}{2} e^{-\lambda/ t_1}\right)}e^{-\lambda(t-t_1)}
%
\end{align} 
Since the trend function of $X(t)$ incorporates a sinusoidal component, we refer to the resulting model as a Sine-Like stochastic diffusion process. The variance of the process is 
\begin{align*}
Var[X(t)]&=E[X^2(t)]-(E[X(t)])^2\\
&=x^2_{1}\left(\frac{t}{t_1}\right)^{4}\frac{\sin^2\left(\frac{\pi}{2} e^{-\lambda/ t}\right)}{\sin^2\left(\frac{\pi}{2} e^{-\lambda/ t_1}\right)}e^{-2\lambda(t-t_1)}\left(e^{\sigma^2(t-t_1)}-1\right).
\end{align*}

\section{Inference on the process}\label{MLsection}
In this section, we examine the ML estimation of the parameters of the model from which we can obtain, from Zenha’s theorem \cite{Zehna}, the corresponding aforementioned estimated trend function and the conditional estimated trend function.
\subsection{ML estimates}
As long as we obtain the explicit expression of the TPDF of the process, we can estimate the two parameters involved in the drift and diffusion functions using the ML method. Let us consider a discrete sampling of the process $x(t_1),x(t_2),\dots,x(t_n)$ at times $t_1,t_2,\dots,t_n=T$. For simplicity, put $t_{j+1}-t_{j}=h$ and  use $x_{i}$ to refer to $x(t_i)=x_i$. Assuming the initial condition $P(X(t_1)=x_1)=1$. The likelihood function of the parameter $\boldsymbol\theta=(\lambda,\sigma^{2})^{T}$ can be obtained from equation (\ref{pdf}) as

\begin{align*}
&\mathcal{L}(\boldsymbol\theta)=\underset{j=1}{\overset{n-1}{\prod}} f(x_{j+1},t_{j+1}|x_{j},t_{j}),\\
&=\underset{j=1}{\overset{n-1}{\prod}}\frac{1}{x_{j+1}}\left[2\pi\sigma^2(t_{j+1}-t_{j}) \right]^{\frac{-1}{2}}\exp\big(\frac{-[\log(x_{j+1})-\mu(j,j+1,x_{j})]^2}{2\sigma^2(t_{j+1}-t_{j})}\big)\\
%
%
%
&=\underset{j=1}{\overset{n-1}{\prod}}\frac{1}{x_{j+1}}\left[2\pi\sigma^2h \right]^{-1/2}
\exp\left(-\frac{\left[\mathcal{H}_{\lambda,j}+ \frac{\sigma ^2 }{2}h\right]^2}{2\sigma^2h}\right),
\end{align*}

where \[\mathcal{H}_{\lambda,j}=\log(\frac{x_{j+1}}{x_j})-2\log(\frac{t_{j+1}}{t_j})+\lambda h-\log\left(\frac{\sin\left(\frac{\pi}{2} e^{-\lambda/ t_{j+1}}\right)}{\sin\left(\frac{\pi}{2}e^{-\lambda/ t_j}\right)}\right).\]
The log-likelihood equation is
\begin{align}\label{LogLikelihood1}
&\boldsymbol\ell(\lambda,\sigma^{2})
=-\frac{n-1}{2}\log(2\pi h)-\frac{n-1}{2}\log(\sigma ^2)-\underset{j=1}{\overset{n-1}{\sum}}\log(x_{j+1})\nonumber\\
&\hspace{1.8cm}-\frac{1}{2\sigma ^2 h}\underset{j=1}{\overset{n-1}{\sum}}\left[\mathcal{H}_{\lambda,j}+ \frac{\sigma ^2 }{2}h\right]^2.
\end{align}
The log-likelihood function $\boldsymbol\ell(\lambda,\sigma^{2})$ can be maximized by solving the nonlinear likelihood equation obtained by
differentiating with respect to the parameter vector $\boldsymbol\theta=(\lambda,\sigma^{2})^{T}$. The first-order partial derivatives of $\boldsymbol\ell(\lambda,\sigma^{2})$ are given by
 \begin{align}
\frac{\partial\boldsymbol\ell(\lambda,\sigma^2)}{\partial \lambda}&=\frac{-1}{\sigma^2 h}\underset{j=1}{\overset{n-1}{\sum}}(\mathcal{H}_{\lambda,j}+ \frac{\sigma^2 h}{2})\Big[h+\frac{\pi e^{-\lambda/t_{j+1}}\cot\left(\frac{\pi}{2}e^{-\lambda/t_{j+1}}\right)}{2 t_{j+1}}\Big]\nonumber\\
&-\frac{1}{\sigma^2 h}\underset{j=1}{\overset{n-1}{\sum}}(\mathcal{H}_{\lambda,j}+ \frac{\sigma^2 h}{2})\Big[-\frac{\pi e^{-\lambda/t_{j}}\cot\left(\frac{\pi}{2}e^{-\lambda/t_{j}}\right)}{2 t_{j}}\Big]=0.\label{withrespectToAlpha1}\\
\frac{\partial\boldsymbol\ell(\lambda,\sigma^2))}{\partial \sigma^{2}}&=-\frac{n-1}{2\sigma ^2 }+\frac{1}{2\sigma ^4 h}\underset{j=1}{\overset{n-1}{\sum}}\mathcal{H}_{\lambda,j}^2-\underset{j=1}{\overset{n-1}{\sum}}\frac{h}{8}=0,\label{withrespectToSigma2}
\end{align}

Let $S(\boldsymbol\theta) = (\partial\boldsymbol\ell(\boldsymbol\theta)/\partial\lambda, \partial\boldsymbol\ell(\boldsymbol\theta)/\partial\sigma^{2})^{T}$ be the score function. The MLE of $\widehat{\boldsymbol\theta}=(\widehat{\lambda},\widehat{\sigma}^{2})$ can be obtained by solving the system of equations $S(\boldsymbol\theta) = 0$. Since closed-form solutions are not available, numerical methods are required to compute these estimates. From equation (\ref{withrespectToSigma2}) we obtain 
\begin{align}\label{OptimalSigma2}
\frac{{\hat{\sigma}}^2}{2}=\frac{1}{h}\left[\left(1+\frac{1}{n-1}\underset{j=1}{\overset{n-1}{\sum}}\mathcal{H}_{\hat{\lambda},j}^2\right)^{1/2}-1\right.]
\end{align}
Then, by substituting the expression of ${{\hat{\sigma}}^2}/{2}$ in \eqref{OptimalSigma2} into Equation \eqref{withrespectToAlpha1}, we obtain the estimator $\hat{\lambda}$ from the following nonlinear equation:
\begin{align}
&\underset{j=1}{\overset{n-1}{\sum}}\Big(\mathcal{H}_{\hat{\lambda},j}+\frac{{\hat{\sigma}}^2 h}{2}\Big)\Big[h+\frac{\pi e^{-\hat{\lambda}/t_{j+1}}\cot\left(\frac{\pi}{2}e^{-\hat{\lambda}/t_{j+1}}\right)}{2 t_{j+1}}\Big]\nonumber\\
&\hspace{2cm}-\underset{j=1}{\overset{n-1}{\sum}}\Big(\mathcal{H}_{\hat{\lambda},j}+\frac{{\hat{\sigma}}^2 h}{2}\Big)\Big[\frac{\pi e^{-\hat{\lambda}/t_{j}}\cot\left(\frac{\pi}{2}e^{-\hat{\lambda}/t_{j}}\right)}{2 t_{j}}\Big]=0.\label{m1}
\end{align}
\subsection{Estimated trend functions and confidence bounds}\label{Confidence Bound}
We provide estimates for the conditional mean and the mean of the process by replacing the parameters in equations (\ref{ECTF1}) and  (\ref{ETF1}) by their ML estimators, due to the invariance property of the ML estimates, see for example, Theorem [5.28, 308] in \cite{Schervish}. Let $\widehat{\lambda}$ and $\widehat{\sigma}^{2}$ be the ML estimates of $\lambda$ and $\sigma^{2}$ respectively, then the estimated conditional mean of the process (ECMF) is
\begin{align}
\widehat{E}(X(t)|X(s))=x_s\left(\frac{t}{s}\right)^{2}\frac{\sin\left(\frac{\pi}{2} e^{-\hat{\lambda}/ t}\right)}{\sin\left(\frac{\pi}{2} e^{-\hat{\lambda}/ s}\right)}e^{-\hat{\lambda}(t-s)}
\end{align} 
Similarly, the estimated mean function (EMF) of the process is 
\begin{align}\label{EMF}
\widehat{E}(X(t)|X(t_1))=x_{1}\left(\frac{t}{t_1}\right)^{2}\frac{\sin\left(\frac{\pi}{2} e^{-\hat{\lambda}/ t}\right)}{\sin\left(\frac{\pi}{2} e^{-\hat{\lambda}/ t_1}\right)}e^{-\hat{\lambda}(t-t_1)},
\end{align} 
taking into account the assumption $P(X(t_{1})=x(t_{1}))=1.$ In addition, we can obtain a confidence band for the CMF and MF of the process, using the procedure described in \cite{Katsamaki}. From Equation \eqref{pdf}, we have that for $t>s$, $X(t)|X(s)$ follows $\Lambda_1(\mu(s,t,x_{s}), \sigma^2 (t-s))$. Therefore, we have that
\begin{align*}
Z&=\frac{\ln(X(t))-\mu(s,t,x_{s})}{\sigma\sqrt{t-s}}\sim \mathcal{N}(0,1).
\end{align*}
Consequently, a $(1-\alpha)100\%$ confidence band for $z$ is determined by $P(-Z_{\alpha/2}\leq Z\leq Z_{\alpha/2})=1-\alpha,$ for $\alpha\in(0,1).$ and hence,
\begin{align*}
\mathbb{P}\left[-Z_{\frac{\alpha}{2}}\leq\frac{\mathcal{H}_{\lambda,t}+\frac{\sigma ^2 }{2}(t-t_1)}{\sigma\sqrt{t-t_1}}\leq Z_{\frac{\alpha}{2}}\right]\approx 1-\alpha.
\end{align*}
here \[\mathcal{H}_{\lambda,t}=\log\left(\frac{x(t)}{x_1}\right)-2\log(t/t_1)+\lambda(t-t_1)-\log\left(\frac{\sin\left(\frac{\pi}{2} e^{-\lambda/ t}\right)}{\sin\left(\frac{\pi}{2}e^{-\lambda/ t_1}\right)}\right).\]

From this, we can obtain a $(1-\alpha)100\%$ confidence bound (CB) for $x(t)$ as
\begin{equation}\label{xConfidenceBound}
x_{lower}(t)\leq x(t)\leq x_{upper}(t)
\end{equation}
where,
\begin{align*}\label{xLowerUpper}
&x_{lower}(t)=x_1e^{-Z_{\alpha/2}\sigma\sqrt{t-t_1}+2\log(\frac{t}{t_1})-\lambda(t-t_1)+\log\left(\frac{\sin\left(\frac{\pi}{2} e^{-\lambda/ t}\right)}{\sin\left(\frac{\pi}{2}e^{-\lambda/ t_1}\right)}\right)- \frac{\sigma ^2 }{2}(t-t_1)}\nonumber\\
 &x_{upper}(t)=x_1e^{Z_{\alpha/2}\sigma\sqrt{t-t_1}+2\log(\frac{t}{t_1})-\lambda(t-t_1)+\log\left(\frac{\sin\left(\frac{\pi}{2} e^{-\lambda/ t}\right)}{\sin\left(\frac{\pi}{2}e^{-\lambda/ t_1}\right)}\right)- \frac{\sigma ^2 }{2}(t-t_1)}.
\end{align*}



\section{Applications}\label{Application}
This section focuses on giving an application of our model to real data, where we give an application in real data for the electricity production from natural gas sources ($\%$ of total) in the United States during the period 1990 to 2023. Raw data are shown in Table \ref{data}; these data are annualized and available in the World Bank's database. We compare our model with an existing model, Gompertz, which is suitable for modelling increasing trend data.


\subsection{Application to electricity production from natural gas sources in the United States ($\%$ of total)}\label{Electricity Production}
Using data from 1990 to 2021, we train the model and estimate its parameters: $\hat{\lambda}=-0.03828096$, and $\hat{\sigma}=  0.0673062$. To forecast electricity production from natural gas sources in the United States ($\%$ of total)  for 2022 and 2023, we apply EMF and ECMF, replacing the estimated parameters in (\ref{ETF1}) and (\ref{ECTF1}), respectively. See Table \ref{InfantsETFandECTF}. Figure \ref{ECTFfigureInfantDeath} shows the
performance of the stochastic SL diffusion process for our forecasts. The left panel displays the observed data, estimated trend function (EMF), and estimated confidence bands, illustrating how well the SL process fits the current data and makes accurate predictions. The right panel also shows the conditional trend function estimation, forecasts, and estimated confidence bands. All computations in this paper were performed using R software \cite{R2016}.
  
\begin{table}[H]
\caption{Electricity production from natural gas sources in the United States.}
\begin{center}
\small
\begin{tabular} {  p {.60 cm}  p {10.56 cm} }
    \hline
    Year&1990 \hspace{1cm} 1991\hspace{1.2cm} 1992 \hspace{1cm} 1993\hspace{1cm} 1994 \hspace{1cm} 1995 \\
     Data&11.85815 \hspace{0.1cm}12.27868 \hspace{.3cm} 12.97155\hspace{0.3cm}  12.92770\hspace{.2cm} 14.16742\hspace{0.5cm}14.76346\\
     \hline
    Year&1996 \hspace{1cm} 1997\hspace{1.2cm} 1998\hspace{1cm} 1999\hspace{1cm} 2000 \hspace{1cm} 2001 \\
     Data&13.02046 \hspace{0.1cm} 13.67453  \hspace{.3cm}14.57905 \hspace{.1cm}  14.93086\hspace{0.4cm}15.65117   \hspace{0.4cm}17.07274\\
      \hline
      Year&2002 \hspace{1cm} 2003\hspace{1.2cm} 2004 \hspace{1cm} 2005\hspace{1cm} 2006 \hspace{1cm} 2007 \\
     Data&17.58605 \hspace{0.2cm}16.41918   \hspace{.3cm}17.52281  \hspace{0.1cm} 18.22920\hspace{0.4cm}19.59561  \hspace{0.4cm}21.03976 \\
      \hline
      Year& 2008 \hspace{1cm} 2009\hspace{1.2cm} 2010\hspace{1cm} 2011\hspace{1cm} 2012 \hspace{1cm} 2013 \\
     Data&20.83612\hspace{0.35cm}22.67735  \hspace{.15cm} 23.24735\hspace{0.3cm}  24.03179\hspace{0.4cm}29.47220 \hspace{0.2cm} 26.90093\\
      \hline
      Year& 2014 \hspace{1cm} 2015\hspace{1.2cm} 2016\hspace{1cm} 2017\hspace{1cm} 2018 \hspace{1cm} 2019 \\
     Data&26.75654 \hspace{0.1cm}  31.79336 \hspace{.15cm} 32.81091 \hspace{0.15cm}  31.20787\hspace{0.25cm} 34.09804\hspace{0.35cm}  37.33868\\
      \hline
      Year& 2020 \hspace{1cm} 2021\hspace{1.2cm} 2022\hspace{1cm} 2023\\
     Data&39.43958 \hspace{0.2cm}37.35339  \hspace{.3cm}38.71585 \hspace{0.25cm}41.90800 \\
      \hline
    \end{tabular} \label{data}
    \normalsize
   \end{center}
   \end{table}
   \normalsize

\begin{table}[H]
\caption{Predictions with EMF and ECMF of the process.}\label{InfantsETFandECTF}
\begin{center}
\begin{tabular} {  p {.7 cm} p {9 cm} }
    \hline
Year&\hspace{1cm} Real Data \hspace{1cm} EMF \hspace{1.7cm}ECMF \\
    \hline
   2022& \hspace{0.8cm} 38.71585 \hspace{01.2cm}41.67541 \hspace{01cm}38.84946\\
   2023&\hspace{0.8cm} 41.90800  \hspace{01.05cm}  43.34456 \hspace{0.9cm} 40.26646 \\
         \hline
    \end{tabular}
   \end{center} 
   \end{table}
 \begin{figure}[H]
 \begin{center}
    \includegraphics[width=8cm,height=5.5cm]{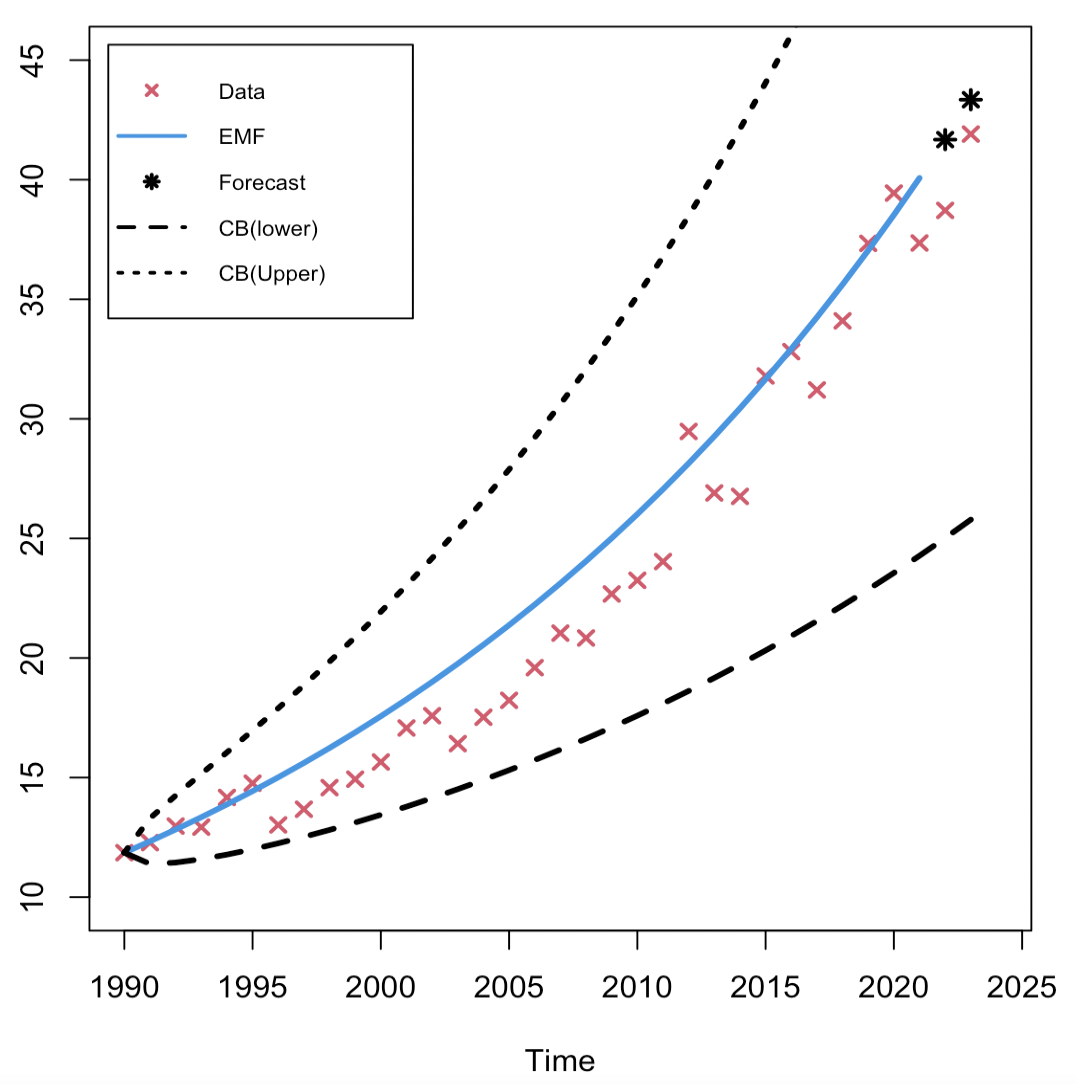}
    \includegraphics[width=8cm,height=5.5cm]{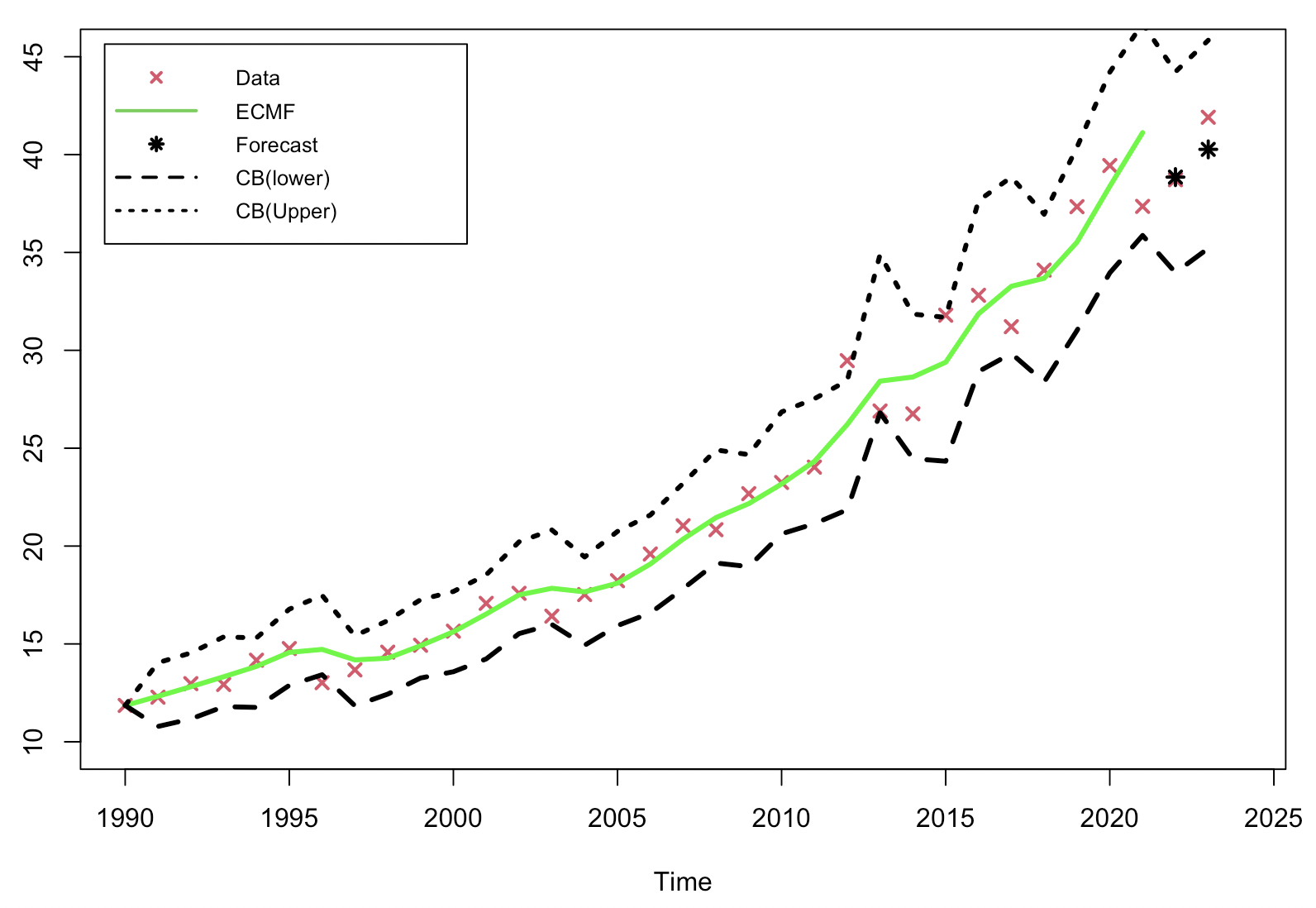}
\caption{Observed data, EMF, ECMF, and the forecasted values.}\label{ECTFfigureInfantDeath}
 \end{center} 
\end{figure}
We compare the proposed diffusion process with the Gompertz process \cite{RR209b} using the existing data. Table \ref{Comparing Models} shows the estimated parameters and the AIC criterion. The SL has a lower AIC and thus excels the Gompertz diffusion process. 
\begin{table}[H]
\caption{SL, Gompertz: Parameters and AIC.}\label{Comparing Models}
\begin{center}
\small
\begin{tabular} {  p {1.6cm}  p {9.5 cm} }
      \hline
Model&\hspace{1cm}$ \lambda$ \hspace{2.3cm} $\beta$  \hspace{2cm}$\sigma$\hspace{2.2cm}AIC\\
    \hline
SL& -0.03828096   \hspace{01.3cm}NA \hspace{1.1cm} 0.06730620 \hspace{0.65cm} 112.3892\\ 
 Gompertz&0.060000022  \hspace{0.3cm}0.006881754\hspace{0.7cm} 0.067492691\hspace{0.6cm} 114.3477\\
         \hline
    \end{tabular}
   \end{center} 
   \end{table}
   \normalsize
 Figure \ref{ETFInfantDeathComparing} shows the fits made using the methods mentioned in Table \ref{Comparing Models}.
  \begin{figure}[hbt!]
 \begin{center}
\includegraphics[width=11cm,height=7.5cm]{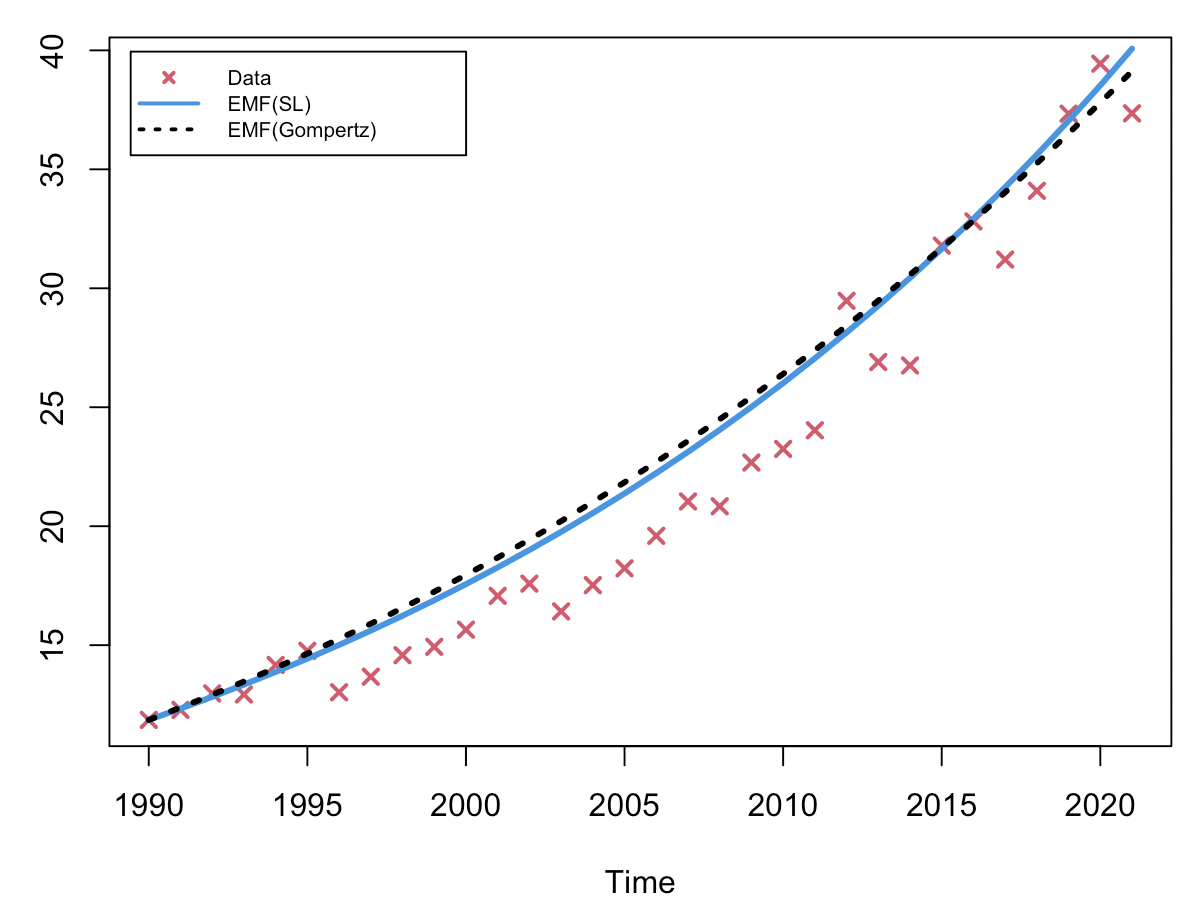}
   \caption{Observed data, EMF using SL, and Gompertz diffusion processes. }\label{ETFInfantDeathComparing}
 \end{center} 
\end{figure}

\subsection{Goodness of fit}\label{Goodness of fit}
Forecast accuracy can be measured using metrics such as mean absolute error (MAE), the root mean square error (RMSE), and the mean absolute percentage error (MAPE). 
We consider the observed data as $x(t)$, the predicted data $\widehat{x}(t)$ which is obtained by substituting the estimated parameter $\hat{\lambda}= -0.03828096$ in equation (\ref{EMF}), and the difference between them $x(t)-\widehat{x}(t)$ over the years 1990-2021. We calculate the mean absolute error (MAE), the root mean square error (RMSE), and the mean absolute percentage error (MAPE) as follows
\begin{align*}
{\text{MAE}}&=\frac{1}{32}\sum_{i=1}^{32}{ |  x(t_i)-\widehat{x_i}(t) |}=1.718274,\\
{\text{RMSE}}&=\sqrt{\frac{1}{32}\sum_{i=1}^{32}{ |  x(t_i)-\widehat{x_i}(t) |^2}}=66.1315,\\
{\text{MAPE}}&=\frac{1}{32}\sum_{i=1}^{32}{\frac{ |  x(t_i)-\widehat{x_i}(t) |}{ x(t)}}\times 100=7.559711.
\end{align*}
The value obtained for MAPE is less than 10, and this indicates that we obtained a high-accuracy prediction according to \cite{Lewis}.
   
 \subsection{Simulation}\label{simulation}
 The sample trajectories were simulated using Equation (\ref{exa}) with values of $\alpha$, $\sigma$, and $x_1$ that approximate the values of these parameters in the real example in the application on which this study is conducted in section \ref{Electricity Production}. Ten trajectories were generated, each containing 500 values. 
  Figure \ref{SimulationAndMF} shows the simulated trajectories of the SL process versus to MF for the particular case of  ${\lambda}=-0.03828096$, ${\sigma}=  0.0673062$, $t_1=1990$, $\Delta t=0.066$, $x_{t_1}=11.85815$ which respectively correspond to values close to those obtained in the study of $X(t)$. We used the same algorithm as in \cite{Safa1}  to simulate the trajectories of the SL process.    
  \begin{figure}[hbt!]
 \begin{center}
\includegraphics[width=11cm,height=6.5cm]{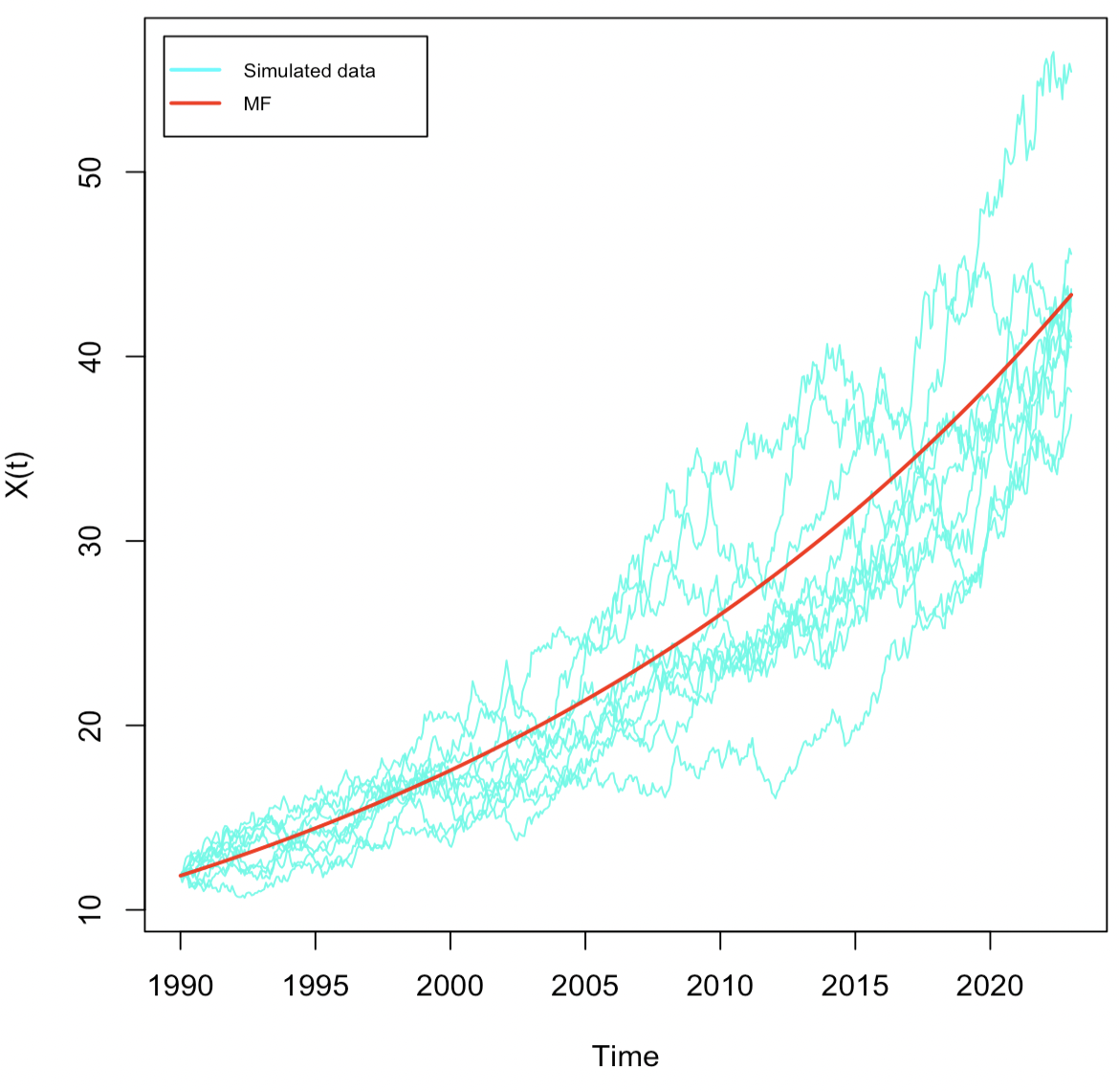}
   \caption{ Sine-like (SL) process simulated along with the mean function (MF).}\label{SimulationAndMF}
 \end{center} 
\end{figure}
\section{Conclusions}\label{Conclusion}
 Based on the results obtained (see Table \ref{InfantsETFandECTF}, Figure \ref{ECTFfigureInfantDeath}), we conclude that modeling the U.S. data on electricity generated from natural gas ($\%$ of total) over the period 1990–2021 using the novel Sine-Like model- estimated according to the procedure outlined in Section \ref{MLsection}- yields a high level of accuracy in the historical trend and produces reliable medium-term forecasts for 2022–2023. The conditioned trend function provides much better descriptions and forecasts than the basic trend function. Furthermore, a comparison with the Gompertz model using the same dataset and sample period demonstrates that the sine-like stochastic diffusion model achieves superior overall performance.\\ 
Future work could build on this study by incorporating external time-varying factors-such as natural gas prices, renewable energy adoption, or seasonal demand- directly into the drift and/or diffusion terms of the SDE. This would require moving from a univariate to a multivariate stochastic framework, such as coupled SDE systems or stochastic regression models.

\end{document}